\documentclass[12pt]{article}
\usepackage{amsfonts,latexsym,amsthm,amssymb,amsmath,amscd,euscript}
\usepackage{framed}
\usepackage{pgfplots}
\usetikzlibrary{positioning, arrows.meta}
\usepgfplotslibrary{fillbetween}
\pgfplotsset{compat=1.8}
\usepackage{array}
\usepackage{verbatim}
\usepackage{lscape}
\usepackage{natbib}
\usepackage[shortlabels]{enumitem}
\usepackage{fullpage}
\usepackage{rotating}
\usepackage{soul}
\usepackage{graphicx}
\usepackage[bottom]{footmisc}
\usepackage{multirow}
\usepackage{caption}
\usepackage{subcaption}

\usepackage{tikz}
\usepackage{hyperref}
\hypersetup{colorlinks=true,citecolor=blue,urlcolor =black,linkbordercolor={1 0 0}}
\usepackage[margin=1in]{geometry}
\usepackage{bbm}
\usepackage{ stmaryrd }
\usepackage{listings}
\usepackage{float}
\usepackage{dcolumn}

\allowdisplaybreaks[1]
\newtheorem{theorem}{Theorem}
\numberwithin{theorem}{section}
\newtheorem{innercustomgeneric}{\customgenericname}
\providecommand{\customgenericname}{}
\newcommand{\newcustomtheorem}[2]{%
	\newenvironment{#1}[1]
	{%
		\renewcommand\customgenericname{#2}%
		\renewcommand\theinnercustomgeneric{##1}%
		\innercustomgeneric
	}
	{\endinnercustomgeneric}
}

\newcustomtheorem{customthm}{Theorem}
\newcustomtheorem{customlemma}{Lemma}
\newcustomtheorem{customprop}{Proposition}
\newcustomtheorem{customcor}{Corollary}
\newtheorem{proposition}[theorem]{Proposition}
\newtheorem{lemma}[theorem]{Lemma}
\newtheorem{corollary}[theorem]{Corollary}

\theoremstyle{definition}
\newcustomtheorem{customdefn}{Definition}
\newtheorem{defn}[theorem]{Definition}
\newtheorem{assumption}[theorem]{Assumption}

\theoremstyle{remark}
\newtheorem*{remark}{Remark}

\newcommand{\post}{\mathcal{P}}

\newcommand{\1}{\mathbbm{1}}

\newcommand{\ut}{\underline{\tau}}

\newcommand{\E}{\mathbb{E}}

\renewcommand{\P}{\mathbb{P}}

\renewcommand{\b}{\beta}
\renewcommand{\t}{\theta}
\newcommand{\g}{\gamma}
\renewcommand{\k}{\kappa}

\title{Average Profits of Prejudiced Algorithms}
\author{David J. Jin\stepcounter{footnote}\thanks{Department of Economics, University of Pennsylvania; \href{mailtO:jindavid@sas.upenn.edu}{jindavid@sas.upenn.edu}. I am deeply indebted to Hanming Fang for his supervision and guidance, as well as his encouragement and criticism. I also extend my gratitude to Jere Behrman for all of the helpful suggestions since the inception of this project. Special thanks go to Kevin He, Jonathan Roth, Keunsang Song, and two anonymous referees. All errors are mine.}}
\date{}
\begin{document}
	
	\maketitle
	
	\begin{abstract}
		We investigate the level of success a firm achieves depending on which of two common scoring algorithms is used to screen qualified applicants belonging to a disadvantaged group. Both algorithms are trained on data generated by a prejudiced decision-maker independently of the firm. One algorithm favors disadvantaged individuals, while the other algorithm exemplifies prejudice in the training data. We deliver sharp guarantees for when the firm finds more success with one algorithm over the other, depending on the prejudice level of the decision-maker.
		\par\null\par\noindent\textit{Keywords: algorithmic fairness, discrimination, bias reversal, profit}
	\end{abstract}
	
	\pagebreak
	
	\section{Introduction}
	Algorithms are increasingly being used to assist or entirely replace human predictions on important outcomes. When designed carelessly, algorithms are susceptible to disparate outcomes, based on unacceptable characteristics like race, sex, or religion. We focus on a scenario where a firm uses a predictive algorithm to hire or reject applicants belonging to a disadvantaged group. The generality of our framework allows for extensions to similar environments such as loan applications, police stops, or parole decisions. We suppose the algorithm is trained on a data set generated by a human decision-maker prejudiced against the disadvantaged group. Then, we determine the level of success firms experience while existing in environments of varying levels of prejudice. We measure the success of a firm using \textit{average profits}, allowing us to determine which algorithm yields the highest payoff when screening for qualified applicants. We examine two of the most accessible and important algorithms denoted by $s_1$ and $s_2$. Algorithm $s_1$ ``reverses'' the prejudice embedded in the training data by favoring applicants belonging to the disadvantaged group. Algorithm $s_2$ ``bakes in'' the prejudice embedded in the training data by assigning lower scores to disadvantaged applicants. These two algorithms are analyzed by \cite{rr}: we say that $s_1$ exhibits ``bias reversal'' while $s_2$ exhibits ``bias in, bias out.'' To be precise, $s_1$ is an algorithm trained on the set of applicants accepted by the human decision-maker and predicts the outcome of interest, i.e., qualification. On the other hand, $s_2$ is trained on the entire set of applicants and predicts whether the human decision-maker accepts or rejects.
	
	Our definition of discrimination is driven by classical ideas in labor economics. We suppose the firm is a \textit{statistical} discriminator, while the data set is generated by both a statistical and \textit{taste-based} discriminating human decision-maker. In reality from a policy perspective, it is difficult to track whether a human decision-maker is taste-based discriminator, precisely because scoring mechanisms and cutoffs are computed mentally. On the other hand, when a firm makes use of an algorithm, concrete scores are output in some software. The explicit documentation of algorithmic scores restricts cases of blatant taste-based discrimination, which is likely restricted by some external oversight agency, i.e., government.
	
	It is unclear with which of the two algorithms the firm is more profitable. Indeed, our framework summarizes a modern challenge that firms face: which is the ``best'' algorithm when a firm is dealt a prejudiced training data set? Moreover, does the best algorithm change depending on the extent of prejudice of the prior era? We consider the interval of possible levels of bias the firm may take: on one extreme end, the human decision-maker hires all disadvantaged applicants, and on the other end, the human decision-maker rejects all disadvantaged applicants. We find that firms living in a sufficiently high-discrimination environment, on average, find more success with the algorithm that transmits bias as time goes on. Conversely, firms in a sufficiently low-discrimination environment are more successful with the algorithm that reverses bias.\footnote{A ``low-discrimination'' environment is intended to mean when the human decision-maker maximally \textit{favors} disadvantaged applicants.} A casual restatement of our main result is as follows.
	
	\setcounter{subsection}{1}
	
	\begin{theorem}[Theorem \ref{main2}, Corollary \ref{gprofit}]
		There exists a unique level of prejudice where the firm is equally profitable with the two algorithms. For any level of lesser prejudice, the firm is more successful with ``bias reversal.'' For any level of greater prejudice, the firm is more successful with ``bias in, bias out.''
	\end{theorem}
	
	It is crucial to note that the firm's success need not be related to the human decision-maker's success. For the sake of full generality, we may even suppose that the human decision-maker and the firm are independent, in the sense that the human decision-maker may predate the firm itself.\footnote{See the beginning of Section \ref{sec:3} for concrete details.} Thus, our results may characterize impartial firms that exist in environments of high prejudice that manifests in the form of a data set. To further emphasize that the human decision-maker and the firm may be orthogonal, our framework allows the payoffs from hiring a qualified or unqualified applicant to differ between the human decision-maker and the firm, suggesting that they are separate entities. Of course, if these payoffs are equal, one may imagine that the human decision-maker is a hiring manager or interviewer belonging to the firm.
	
	A brief summary of our framework is as follows. We consider hiring decisions in two stages: in the first stage, a human decision-maker hires or rejects disadvantaged applicants. Upon hiring an applicant, the human decision-maker observes his or her true level of qualification. At the conclusion of the first stage, the human decision-maker compiles all hiring data, which is then made available to the firm in the second stage. Importantly, this data set suffers from \textit{selective labels} as described by \cite{slp} because the true ability of rejected applicants is unobserved. In the second stage, the firm uses this data set to construct one of the two aforementioned algorithms. Using this algorithm, the firm scores and accepts or rejects future applicants, after which profits become observed.
	
	This paper is most closely related to the intersection of the study of discrimination and algorithmic bias, often referred to as \textit{algorithmic fairness}. In particular, we implement the two theories of taste-based and statistical discrimination in the setting of algorithmic predictions. \cite{beck} defines taste-based discrimination as the act of harboring a ``taste for discrimination'' against a certain group, implying a positive payoff in favor of discriminating. \cite{fang} generally refers to statistical discrimination as the act of a human decision-maker using observed characteristics of an individual in place of unobserved yet outcome-relevant characteristics. The seminal works of \cite{phelps} and \cite{arrow} led to influential papers involving models of signal extraction in the labor market. For example, \cite{lund} incorporates the notion of skill investments in Phelps' model. \cite{cl} exploits this idea along with Arrow's model to show that two \textit{ex ante} identical groups may end up in different Pareto ranked equilibria. Adjacent areas of research investigate discrimination arising from inter-group interactions (\cite{moro, mailath, fang2001}). These models can be used to inform policy decisions. For instance, \cite{fang2006,chung, chan} and others apply the models of discrimination to evaluate the effects of affirmative action, where governments institute policies in favor of disadvantaged groups. This paper aims to fit a traditional model of discrimination into the context of algorithmic bias. An empirical study conducted by \cite{propub} reveals that a risk assessment algorithm used by the court system in Broward County, Florida is more likely to falsely flag black defendants as future criminals. Many subsequent analyses reveal that under general circumstances, it is impossible to achieve multiple metrics of fairness simultaneously (\cite{ito, hardt, fpdi, corb}). These findings suggest that common standards of algorithmic fairness are challenging to achieve, if possible at all. Given this difficulty, one is naturally interested in examining the behavior of algorithms in the presence of unfairness in order to devise empirical strategies to reduce bias as much as possible. In an effort to pin down the source of such algorithmic prejudice, an influx of studies spearheaded by econometricians and computer scientists propose a number of explanations. The notion of ``bias in, bias out'' is a suspect of generating algorithmic bias, where if an algorithm is trained on a data set generated by prejudiced human decision-makers, the algorithm would reflect these biases (\cite{sol, bibo}). Another candidate behind sowing bias into algorithmic predictions is known as the \textit{selective labels problem}: a type of missing data problem with non-randomly selected samples, where a human decision-maker determines which outcomes of observations are observed (\cite{slp, charfair, rr}). Since the outcomes of observations that are not admitted are not observed by the researcher, making accurate predictions becomes challenging when restricted to the usual econometric tools. A particular driver of this paper, \cite{rr} presents a surprising result where an algorithm reverses the prejudice of the human decision-maker under selectively labeled data, a phenomenon they refer to as ``bias reversal.'' That is, the algorithm treats the minority group more favorably, contrary to the notion of ``bias in, bias out.''
	
	The paper proceeds as follows. In Section \ref{sec:3}, we formulate our framework and analyze the two algorithms. In Section \ref{sec:4}, we define average profits and deliver results on firm outcomes. In Section \ref{sec:5}, we display a series of simulations. In Section \ref{sec:6} we conclude. 
	
	\section{A Two-Stage Model}\label{sec:3}
	
	In this section, we develop a mathematical framework for the scenario of a firm hiring applicants in two stages. The first stage involves a prejudiced human decision-maker scoring applicants. In the second stage, the firm builds a predictive algorithm by using the data generated in the first stage. We suppose the firm is born without access to an existing algorithm. We allow the human decision-maker to be either independent of or belonging to the firm---in either case, our results are unchanged. Moreover, the joint distribution of the applicant population is allowed to differ between the two stages, so long as our main distributional assumption, Assumption \ref{as1}, still holds. The distinguishing factor between the first and second stages is solely the mechanism by which applicants are scored. While valuable as an extension to this paper, we do not consider competitive interactions between multiple firms. Therefore, we suppose there is a single firm that chooses to accept or reject each applicant from a set.
	
	\subsection{Environment and Information Structure}
	In this section, we detail the environment. Consider a collection of applicants indexed by the set $I:=\{1,\hdots, n\}$. Each applicant is described by a random variable $Q$ where $Q\in\{0,1\}$ denotes qualification. We let $Q=1$ correspond to qualified applicants and $Q=0$ to unqualified applicants. We let $\pi:= \P(Q=1)\in (0,1)$. We limit $I$ to only include disadvantaged applicants because we assume that the firm is a statistical discriminator; while the firm in reality observes two sets of applicants distinguished by potentially contentious variables, e.g., race, we allow the firm to use entirely different algorithms for the two groups. In this sense, in general it may be misleading to consider which algorithms a firm is better off with based on the performance of algorithmic predictions on a pooled group of applicants.
	
	Suppose a human decision-maker chooses to either accept or reject job applications from $I$. The human decision-maker's payoff from an arbitrary applicant $i\in I$ depends on $Q_i$. The human decision-maker receives a positive payoff $x_q>0$ from hiring a qualified applicant. Conversely, the human decision-maker faces a loss $-x_u<0$ if the human decision-maker hires an unqualified applicant. If the human decision-maker does not hire a particular applicant, they receive zero payoff regardless of whether or not the applicant was qualified. We display the payoff structure in Table \ref{tbl:tab1}.
	
	\begin{table}[h]
		\centering
		\begin{tabular}{l|l|c|c|c}
			\multicolumn{2}{c}{}&\multicolumn{2}{c}{Hiring Decision}&\\
			\cline{3-4}
			\multicolumn{2}{c|}{}&Hire&Reject&\multicolumn{1}{c}{}\\
			\cline{2-4}
			\multirow{2}{*}{Applicant}& Qualified & $x_q>0$ & $0$\\
			\cline{2-4}
			& Unqualified & $-x_u<0$ & $0$\\
			\cline{2-4}
		\end{tabular}
		\caption{The human decision-maker's payoff from hiring a qualified or unqualified applicant}\label{tbl:tab1}
	\end{table}
	
	It is critical to note that we allow the human decision-maker to be independent of the firm. That is, the human decision-maker may be screening these applicants for a separate entity---the only link we strictly require between the human decision-maker and the firm is that the firm acquires access to the data set generated by the human decision-maker. To be precise, the firm values qualified and unqualified applicants differently from the human decision-maker: the firm receives a payoff of $v_q$ for hiring a qualified applicant and suffers a cost $-v_u$ for hiring an unqualified applicant. Of course, one may allow $x_q=v_q$ and $x_u=v_u$ if the human decision-maker is either very similar to or belongs to the firm.
	
	The human decision-maker is naturally unable to observe each $Q_i$ directly. Instead, the human decision-maker associates each applicant with a realization of the pair of random variables $(\Theta, \Gamma)$. Intuitively, $\Theta$ represents the concrete data that can be documented (e.g., highest level of education, GPA, years of experience). We let $\Gamma$ represent the undocumented characteristics of an applicant that the human decision-maker extracts. For example, the applicant's demeanor, manner of speaking, and interpersonal abilities are observed by the human decision-maker oftentimes through an interview. However, $\Gamma$ is unable to documented in a data set, due to its qualitative and instinctual nature. Indeed, the comparative statics results described in \cite{rr} exploit the unobserved nature of $\Gamma$.
	
	Next, the joint law $(\Theta, \Gamma\mid Q)$ corresponds to  the joint density functions $h_q(\t, \gamma)$ when $Q=1$ and $h_u(\t, \gamma)$ when $Q=0$. Thus, we allow for the joint law of the signal to depend on the true qualification level. Denote $f_x(\t)$ and $g_x(\gamma)$ the marginal density functions where $x\in \{q,u\}$. We present our main assumption that characterizes the qualifications of applicants through their signal values.
	
	\begin{assumption}[Monotone Likelihood Ratio Property]\label{as1}
		The joint densities satisfy the \textit{strict monotone likelihood ratio property}, where 
		\[
		l(\t, \gamma) := \frac{h_q(\t,\gamma)}{h_u(\t, \gamma)}
		\]
		is jointly, strictly increasing and continuous in $\theta$ and $\gamma$.
	\end{assumption}
	
	\begin{remark}
		In the univariate case, the space is totally ordered. In our case, the partial order is naturally given by $(\t_1,\g_1) < (\t_2,\g_2)$ if and only if $\t_1 < \t_2$ and $\g_1 < \g_2$. It should be noted that the claim need not hold for $\t_1 < \t_2$ but $\g_1 > \g_2$ for example.
	\end{remark}
	
	We opt for the strict version of the property rather than the weak version, since the latter suggests that higher signals may not correspond to being a better signal. Assumption \ref{as1} intuitively suggests that qualified applicants are more likely to receive higher signals over unqualified applicants. Assumption \ref{as1} immediately implies that both $f_q(\t)/f_u(\t)$ and $g_q(\g)/g_u(\g)$ are strictly increasing and continuous in $\t$ and $\g$ respectively.
	
	\subsection{First Stage: Human Prediction}
	We consider the first stage of applications, where applicants are accepted or rejected by means of a human decision-maker. The skeleton of this framework is largely motivated by \cite{cl}. This setting corresponds to traditional job applications where a senior employee or human resources manager interviews applicants and decides which ones to hire. We assume the human decision-maker is Bayesian, in the sense that he updates his prior beliefs on the level of qualification of each applicant by using her signals $\t$ and $\gamma$. The posterior probability $\P(Q=1 \mid \Theta = \t, \Gamma = \g)$ that the applicant is qualified denoted by $\kappa(\t,\gamma)$ follows from Bayes' rule:
	
	\begin{equation}\label{eq:1}
		\kappa(\t,\gamma) := \frac{\pi h_q(\t, \g)}{\pi h_q(\t, \g) + (1-\pi)h_u(\t, \g)}
	\end{equation}
	
	Naturally, the human decision-maker assigns higher posterior probabilities to applicants with greater signals. It is well known that $\kappa(\t, \g)$ is jointly strictly increasing. Next, let $\tau\in \mathbb{R}$ denote the decision-maker's prejudice against applicants. Then, the decision-maker accepts each applicant if and only if the expected payoff exceeds the cutoff:
	
	\begin{equation}\label{eq:2}
		\kappa(\t,\gamma)x_q - [1- \kappa(\t, \gamma)]x_u > \tau.
	\end{equation}
	
	In other words, the human decision-maker is a taste-based discriminator. We may interpret $\tau<0$ as a human decision-maker who favors disadvantaged applicants. When $\tau \geq x_q$, the human decision-maker's taste for discrimination overrides the payoffs from hiring qualified applicants leading to every applicant being rejected regardless of qualification. Conversely, if $\tau \leq -x_u$, the human decision-maker favors applicants to the point where everyone is accepted. To avoid the trivial, we focus on the interval $\tau \in (-x_u, x_q)$. Inserting the expression for posterior probability \eqref{eq:1} into \eqref{eq:2}, we define the indicator $A(\tau)$ which equals 1 if the applicant is accepted and 0 otherwise:
	
	\begin{equation}\label{eq:3}
		A(\tau) := \1\left\{ \frac{h_q(\t,\g)}{h_u(\t,\g)}  > \frac{(1-\pi)(x_u+\tau)}{\pi(x_q - \tau)}\right\}.
	\end{equation}
	
	We frequently denote $A(\tau)$ by simply $A$ for convenience, but it should be understood that $A$ critically depends on the parameter $\tau$. Once the human decision-maker scores and hires or rejects applicants according to the rule \eqref{eq:3}, the hired applicants' true types $Q$ are revealed. The firm then compiles the signals and outcomes for all applicants into a data set $\{D_i\}_{i\in I} = (\t_i, A_i, Q_i)$. Importantly, $D_i$ does not contain $\gamma_i$. Also, $Q_i$ is only observed when $A_i=1$. Otherwise, we treat it as a missing value. Here, we can see that the firm faces selective labels. An extension may discuss a scenario where the true types of rejected applicants are observed. For example, the productivity of rejected applicants may be observed at competing firms.
	
	\subsection{Second Stage: Algorithmic Prediction}
	With the conclusion of the first stage of hires, the firm is now able to construct algorithms using the data set $\{D_i\}_{i\in I}$. The joint distribution of the observed data crucially depends on the value of $\tau$. Since $Q_i$ is only observed for $A_i=1$, different values of $\tau$ lead to different applicants possessing a label.
	
	We introduce a second set of applicants indexed by the set $J:=\{1,\hdots, m\}$. Identically to the first stage, each applicant is described by $Q$, but instead of a human decision-maker personally scoring applicants, the firm who does not observe $\Gamma$ builds an algorithm. In the following, we assume the data-generating process of $(\Theta, \Gamma, Q)$ are identical to the first stage, but this need not be the case. So long as the joint distribution of $\Theta$ and $\Gamma$ conditional on $Q$ satisfy Assumption \ref{as1}, all of our following results hold. This generality allows for the argument that the inherent distribution of applicants changes over time. In any case, we let the joint distributions to be identical to that of the first stage to simplify proofs, but we reiterate that identical results can be derived using small modifications in argumentation with new joint distributions.
	
	 While real-life firms have countless options on how to design predictive algorithms, we abstract away from the estimation problem and prioritize the behavior of the algorithms as $\tau$ changes. We consider two common predictors under squared loss:
	
	\begin{equation*}
		s_1(\t, \tau) := \E[Q\mid \Theta = \t, A(\tau)=1]
	\end{equation*}
	\begin{equation*}
		s_2(\t, \tau) := \E[A(\tau) \mid \Theta = \t].
	\end{equation*}
	
	We frequently denote the two by $s_1$ and $s_2$ for convenience, omitting some or all arguments when appropriate. The key difference between the two algorithms is the outcome of interest. Algorithm $s_1$ predicts whether applicants are qualified within the accepted pool while algorithm $s_2$ predicts whether applicants are accepted. It follows that $s_1$ is trained on selectively labeled data, and $s_2$ is trained on all of $\{D_i\}_{i\in I}$. \cite{rr} show that $s_1$ and $s_2$ exhibit ``bias reversal'' and ``bias in, bias out,'' respectively.
	
	\begin{theorem}[{\cite[Theorem 1 and 3]{rr}}]\label{bin}
		For any $\t$, $s_1(\t, \tau)$ is strictly increasing in $\tau$ and $s_2(\t, \tau)$ is strictly decreasing in $\tau$.
	\end{theorem}
	
	In other words, the algorithm $s_1$ exhibits ``bias reversal,'' where the score of an applicant increases with the human decision-maker's level of taste-based discrimination. Since firms are interested in the true qualification label of applicants, the purpose of $s_1$ is clear. On the other hand, the intuition of $s_2$ does not come immediately. Prior work discuss that the target variable that employers aim to predict widely varies, and often times data constraints akin to selective labels may lead firms to choose future employees based on who has been hired in the past. Moreover, machine learning engineers may simply be tasked by management to find a target variable that makes selection of employees easier, and scores similar to $s_2$ are leading candidates.\footnote{See \cite{evc} and \cite{bo} for a discussion on the choice of target variables in employment settings.} Before we continue, we prove continuity.
	
	\begin{lemma}\label{cont}
		$s_1(\t,\tau)$ and $s_2(\t,\tau)$ are continuous in $\tau$ for any $\t$.
	\end{lemma}
	
	\begin{proof}
		See Appendix \hyperref[AppendixA]{A}.
	\end{proof}
	
	We next evaluate interactions between the two algorithms and compare scores at given values of $\tau$. In particular, there exists a unique level of prejudice where $s_1$ and $s_2$ are equal.
	
	\begin{proposition}\label{prop1}
		For each $\t$, there exists a unique $\tau^\star\in (-x_u,x_q)$ where $s(\t, \tau^\star):= s_1(\t, \tau^\star)=s_2(\t, \tau^\star)$.
	\end{proposition}
	
	\begin{proof}
		See Appendix \hyperref[AppendixA]{A}.
	\end{proof}
	
	\begin{figure}[H]
		\centering
		\captionsetup{justification=centering}
		\resizebox{10cm}{!}{\linespread{1.25}
\pgfplotsset{compat = newest}
\begin{tikzpicture}
\begin{axis}[
xmin = -11, xmax = 11,
ymin = -0.07, ymax = 1,
xtick = {-10, 10},
ytick = {0, 1},
clip = false,
scale = 0.8,
axis lines* = left,
x label style = {at={(axis description cs:0.5,-0.1)},anchor=north}, xlabel={Human Prejudice $\tau$},
y label style={anchor=south}, ylabel={\large Score $s$},
xticklabels={$-x_u$,$x_q$}
]
\addplot[domain = -10:10, restrict y to domain = -10:10, samples = 300, color = black]{((-1.13 / 3.14) * rad(atan(0.4 * x))) + 0.5 };
\addplot[domain = -10:10, restrict y to domain = -10:10, samples = 300, color = black]{((0.83 / 3.14) * rad(atan(0.3 * x + 0.4))) + 0.65};
\addplot[color = black, mark = *, only marks, mark size = 3pt] coordinates {(-1.2, 0.66)};
\addplot[color = black, dashed, thick] coordinates {(-1.2, -0.07) (-1.2, 0.66)};
\addplot[color = black, dashed, thick] coordinates {(-11, 0.33) (11, 0.33)};
\node [below] at (-1.2, -0.07) {$\tau^\star$};
\node [right] at (11, 0.33) { $\mathbb{P}(Q=1\mid \Theta = \theta)$};
\addplot[color = black, dashed, thick] coordinates {(-11, 0) (11, 0)};
\addplot[color = black, dashed, thick] coordinates {(-11, 1) (11, 1)};
\node [below] at (7, 0.83) { $s_1(\theta, \tau)$};
\node [below] at (-7, 0.83) { $s_2(\theta, \tau)$};
\draw[->] (-7, 0.8) to (-6.5, 0.89);
\draw[->] (7, 0.8) to (6.3, 0.92);
\end{axis}
\end{tikzpicture}}
		\caption{Predicted scores from both algorithms $s_1$ and $s_2$ for\\an applicant across $\tau\in (-x_u,x_q)$}\label{fig:fig1}
	\end{figure}
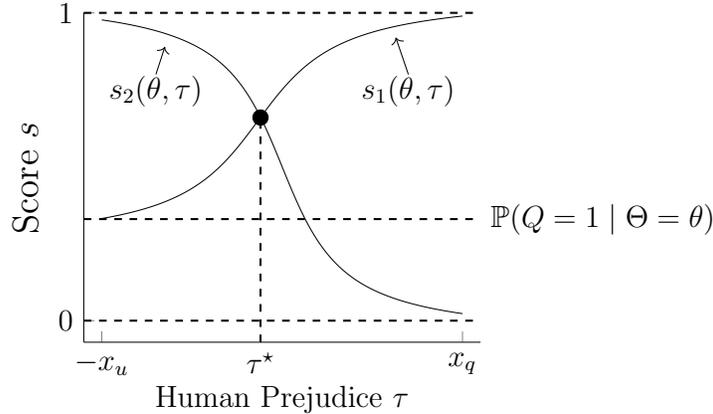
	
	As the human decision-maker becomes more prejudiced, $s_1$ scores eventually become greater than their respective $s_2$ scores. Put differently, for each applicant, there must exist a unique level of prejudice $\tau^\star$ where $s_1$ and $s_2$ are equal. For all levels of prejudice below $\tau^\star$, $s_2$ produces a greater score than $s_1$ and the opposite for all levels of prejudice above $\tau^\star$.
	
	\begin{corollary}\label{greater}
		Fix $\t$. Then, $s_1(\t, \tau) > s_2(\t, \tau)$ if and only if $\tau > \tau^\star$.
	\end{corollary}
	\begin{proof}
		The proof is immediate from evaluating $s_1(\t,\tau)$ and $s_2(\t,\tau)$ for $\tau>\tau^\star$ or $\tau<\tau^\star$ using Theorem \ref{bin}.
	\end{proof}
	
	An equivalent statement of Corollary \ref{greater} is $s_1(\t, \tau)<s_2(\t,\tau)$ if and only if $\tau<\tau^\star$. Intuitively, if the human decision-maker is less (more) prejudiced than the equalizing level of prejudice, the algorithm $s_2$ produces greater (smaller) predictions than $s_1$. Figure \ref{fig:fig1} visually demonstrates the claims of Proposition \ref{prop1}, Corollary \ref{greater}, and Corollary \ref{cor1}.
	
	\begin{corollary}\label{cor1}
		For $\tau \to -x_{u}$, we have $s_1(\t,\tau) \to \E[Q\mid \Theta = \t]$ and $s_2(\t,\tau) \to 1$. For $\tau \to x_{q}$, we have $s_1(\t,\tau) \to 1$ and $s_2(\t,\tau) \to 0$.
	\end{corollary}
	
	\begin{proof}
		The values are computed in the proof of Proposition \ref{prop1}.
	\end{proof}
	
	Corollary \ref{cor1} precisely states the limiting values of $s_1$ and $s_2$ as prejudice approaches its extreme values. That is, $\tau$ approaching $-x_u$ or $x_q$ indicates the extreme cases of partiality or prejudice respectively. To understand the behavior of $s_1$: intuitively, if the human decision-maker hires all applicants, the probability that an applicant is qualified converges to the population's proportion of qualified applicants for that specific $\t$. Conversely, if the human decision-maker's prejudice level approaches $x_q$, only the applicants with the best realizations of $\Gamma$ are accepted, leading probabilities of being qualified to approach 1. The behavior of $s_2$ is more straightforward because it matches the human decision-maker's hiring rule.
	
	% ELABORATE
	In order to analyze with which algorithm the firm finds more success, we define average profit in Section \ref{sec:5}. Regardless of which of the two algorithms a firm chooses, the procedure for computing average profit does not vary. Hence, we arbitrarily denote the applicants' scores as $s$ which may either be $s_1$ or $s_2$. With a slight abuse of notation, we associate the score $s_i$ with applicant $i\in J$ when the algorithm $s$ is arbitrary. After an applicant $i\in J$ is matched with a score $s_i$, the firm's expected payoff from applicant $i$ is given by
	
	\begin{equation}\label{eq:6}
		\mathcal{N}(s_i) := s_iv_q - [1-s_i]v_u
	\end{equation}
	
	where $v_q$ and $v_u$ are the firm-analogues of the human decision-maker's $x_q$ and $x_u$. In \eqref{eq:6}, the first term reflects the gain from making a correct hire, and the second term reflects the loss from making a bad hire. \eqref{eq:6} is precisely the second-stage analogue of the first-stage predicted payoff in \eqref{eq:2}. The only difference is that the firm uses a heuristic where $s(\t)$ is treated as a posterior probability in place of the $\kappa(\t, \g)$. Similar to the decision rule in the first stage \eqref{eq:3}, we define the indicator $A'(\tau)$ which equals 1 if the applicant is accepted and 0 otherwise.
	
	\begin{equation}\label{eq:7}
		A'(s_i) := \mathbbm{1}\left\{s_iv_q - [1-s_i]v_u > 0\right\}
	\end{equation}
	
	In other words, the firm accepts an applicant if their predicted payoffs are positive and rejects otherwise. Note that, unlike the first-stage, we assume that the firm is unable to discriminate based on taste.\footnote{When an algorithm is being used to make decisions, we assume the existence some oversight agency or regulation that restricts a concrete case of taste-based discrimination, as discussed in the introduction.} Once the second round of hires is complete, the firm observes the true types of the hires $Q_i$ and observes profits for each applicant $i$:
	
	\begin{equation}
		\mathcal{P}(s_i) := A'(s_i)\cdot\left(Q_iv_q - [1-Q_i]v_u\right).
	\end{equation}
	
	$\mathcal{P}$ is precisely the realized gain from each applicant based on their true type. If an applicant was not accepted, regardless of $Q_i$ the firm does not make any gains. We can see that the firm in some sense is a ``misspecified Bayesian,'' where $s$ is a misspecified prior because it does not incorporate $\Gamma$ which is unobserved. Figure \ref{fig:fig2} visually illustrates the proceedings of our two-staged model.
	
	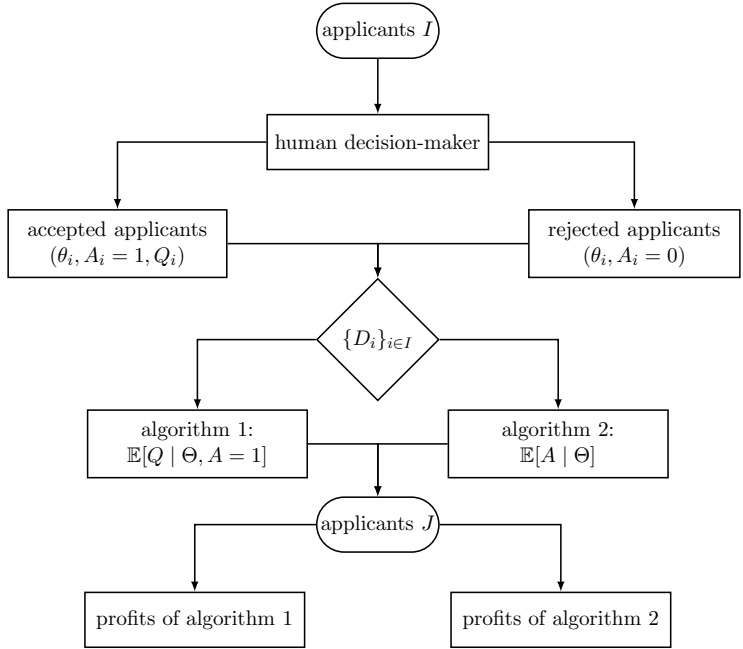
\begin{figure}[h]
		\centering
		\captionsetup{justification=centering}
		% Required packages
\usetikzlibrary{shapes,positioning}
\centering
	\scalebox{0.73}{
		\begin{tikzpicture}[font=\small,thick, node distance=1cm]
			% Place nodes
			\node[draw,
			rounded rectangle,
			minimum width=2.5cm,
			minimum height=1cm] (I) {applicants $I$};
			\node[draw,
			below=of I,
			minimum width=3.5cm,
			minimum height=1cm] (hdm) {human decision-maker};
			\node[draw,
			below left = of hdm,
			minimum width = 3.5cm, minimum height = 1cm] (ra) {\begin{tabular}{c}
					accepted applicants \\ $(\theta_i, A_i = 1, Q_i)$
			\end{tabular}};
			\node[draw,
			below right = of hdm,
			minimum width = 3.5cm, minimum height = 1cm] (rj) {\begin{tabular}{c}
					rejected applicants \\ $(\theta_i, A_i = 0)$
			\end{tabular}};
			
			\node[draw,
			diamond,
			below = of hdm] (d) at (0, -3.5) {$\{D_i\}_{i\in I}$};
			
			\node[draw,
			below left = of d,
			minimum width = 4cm, minimum height = 1cm] (a1) {\begin{tabular}{c}
					algorithm 1: \\
					$\mathbb{E}[Q\mid \Theta, A = 1]$
			\end{tabular}};
			
			\node[draw,
			below right = of d,
			minimum width = 4cm, minimum height = 1cm] (a2) {\begin{tabular}{c}
					algorithm 2: \\ 
					$\mathbb{E}[A\mid \Theta]$
			\end{tabular}};
			
			\node[draw,
			rounded rectangle,
			minimum width=2.5cm,
			minimum height=1cm] at (0, -9) (J) {applicants $J$};
			
			\node[draw,
			below left = of J,
			minimum width = 4cm, minimum height = 1cm] (r1) {profits of algorithm 1};
			
			\node[draw,
			below right = of J,
			minimum width = 4cm, minimum height = 1cm] (r2) {profits of algorithm 2};
			
			\draw[-latex] (I) edge (hdm);
			\draw[-latex] (hdm) -| (ra);
			\draw[-latex] (hdm) -| (rj);
			\draw[-latex] (ra) -| (d);
			\draw[-latex] (rj) -| (d);
			\draw[-latex] (d) -| (a1);
			\draw[-latex] (d) -| (a2);
			\draw[-latex] (a1) -| (J);
			\draw[-latex] (a2) -| (J);
			\draw[-latex] (J) -| (r1);
			\draw[-latex] (J) -| (r2);
		\end{tikzpicture}
	}
		\caption{An illustration of the full timeline of the two-stage model}\label{fig:fig2}
	\end{figure}
	
	\section{Profits}\label{sec:4}
	
	After making hiring decisions using the rule in \eqref{eq:7}, the firm observes profits $\post(s_i)$. In the previous section, we did not make any claims about with which algorithm the firm is more successful. In particular, it is unclear whether the firm would rather use an algorithm that transmits bias or one that reverses bias.
	
	Denote $|J| = m$ and suppose $\Theta$ and $Q$ have the same data generating process specified in Section \ref{sec:3}. As already discussed, we may allow the joint distributions to differ from the first stage, but assuming the same simplifies discussions and proofs. In this section, we let $m$ increase as more applicants apply to the firm. Thus, the notion of profits is clearly defined.
	
	\begin{defn}[Average Profits]
		The \textit{average profit} generated by an algorithm $s$ for a prejudice level $\tau$ is defined as
		\[
		u(s(\tau)) := \frac{1}{m}\sum_{i=1}^m \post(s_{i}(\tau)) = \frac{1}{m}\sum_{i=1}^m \E[A'(s(\tau))\cdot (Qv_q - [1-Q]v_u)].
		\]
	\end{defn}
	
	In Figure \ref{fig:fig3}, we display a plot of profit $\post(s_1(\t, \tau))$ for both a qualified and unqualified applicant over human prejudice. The vertical dotted line signifies the level of $\tau$ such that $\mathcal{N}(s_1(\tau)) = 0$. In other words, for any level of prejudice greater than this $\tau$, the applicant is accepted. Profits jumps upward for $Q=1$ by $v_q$ and downwards for $Q=0$ by $v_u$. Naturally, profits become $v_q>0$ upon hiring a qualified applicant and $-v_u<0$ for an unqualified applicant. It's easy to see that $u(s(\tau))$ is discontinuous at $m$ points at most.\footnote{Profit need not be discontinuous for all $m$ applicants. In particular, if $\E[Q\mid \Theta = \t] > v_u/(v_u + v_q)$, then profit is continuous for that $\t$.}
	
	\begin{figure}[H]
		\centering
		\captionsetup{justification=centering}
		\resizebox{12cm}{!}{% \documentclass[12pt]{article}
% \linespread{1.25}
% \usepackage{pgfplots}
% \pgfplotsset{compat = newest}
% \usetikzlibrary{positioning, arrows.meta}
% \usepgfplotslibrary{fillbetween}
% \usepackage{amsmath}

% \begin{document}
\linespread{1.25}
\pgfplotsset{compat = newest}
\begin{tikzpicture}
\begin{axis}[
title = {(a) $Q=1$},
xmin = -11, xmax = 11,
ymin = -0.07, ymax = 1,
xtick = {-10, 10},
ytick = {0, 0.21, 0.72, 1},
clip = false,
scale = 0.8,
x label style = {at={(axis description cs:0.5,-0.1)},anchor=north}, xlabel={Human Prejudice $\tau$},
y label style={anchor=south}, ylabel={Profit $\post(s_1(\tau))$},
xticklabels={$-x_u$,$x_q$},
yticklabels = {}
]
\draw[dashed] (-1.2, 0.21)--(-1.2, 0.72);
\draw[dashed] (-10, 0.21)--(10, 0.21);
\draw[dashed] (-10, 0.72)--(10, 0.72);
\draw[-, line width = 0.50mm] (-10, 0.21)--(-1.2, 0.21);
\draw[-, line width = 0.50mm] (-1.2, 0.72)--(10, 0.72);
\node [right] at (11.3, 0.72) {$v_q$};
\node [right] at (11.3, 0.21) {$0$};
\node at (-1.2, 0.205) {\textbullet};
\node at (-1.2, 0.715) {$\circ$};
% \node [below] at (7, 0.83) {\small $s^1(\theta, 1, \tau)$};
% \node [below] at (-7, 0.83) {\small $s^2(\theta, 1, \tau)$};
\end{axis}
\end{tikzpicture}
\hspace{11pt}
\begin{tikzpicture}
\begin{axis}[
title = {(b) $Q=0$},
xmin = -11, xmax = 11,
ymin = -0.07, ymax = 1,
xtick = {-10, 10},
ytick = {0, 0.21, 0.72, 1},
clip = false,
scale = 0.8,
x label style = {at={(axis description cs:0.5,-0.1)},anchor=north}, xlabel={Human Prejudice $\tau$},
y label style={anchor=south}, ylabel={Profit $\post(s_1(\tau))$},
xticklabels={$-x_u$,$x_q$},
yticklabels = {}
]
\draw[dashed] (0.3, 0.21)--(0.3, 0.72);
\draw[dashed] (-10, 0.21)--(10, 0.21);
\draw[dashed] (-10, 0.72)--(10, 0.72);
\draw[-, line width = 0.50mm] (-10, 0.72)--(0.3, 0.72);
\draw[-, line width = 0.50mm] (0.3, 0.21)--(10, 0.21);
\node [right] at (11.3, 0.72) {$0$};
\node [right] at (11.3, 0.21) {$-v_u$};
\node at (0.3, 0.715) {\textbullet};
\node at (0.3, 0.205) {$\circ$};
% \node [below] at (7, 0.83) {\small $s^1(\theta, 1, \tau)$};
% \node [below] at (-7, 0.83) {\small $s^2(\theta, 1, \tau)$};
\end{axis}
\end{tikzpicture}
% \end{document}
}
		\caption{Profit generated by algorithm $s_1$ for a\\qualified and unqualified applicant}\label{fig:fig3}
	\end{figure}
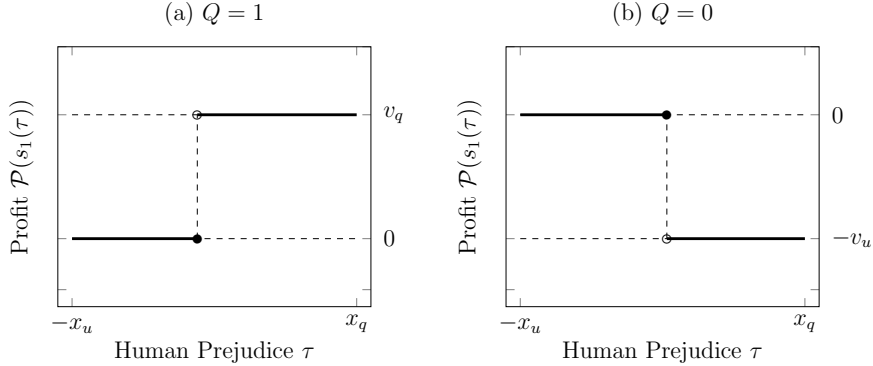

Before we present our main results, we state a mild assumption on the behavior of $s_1(\t)$ to guide our proofs.

	\begin{lemma}\label{s2inc}
	$s_2(\t)$ is strictly increasing in $\t$ for any $\tau\in (-x_u,x_q)$.
\end{lemma}

\begin{proof}
	Recall that $s_2(\t, \tau) = \P(\Gamma > \g_1(\t))$ where $\g_1(\t)$ is defined in \eqref{eq:g1}. As $\t$ increases, we see that $\g_1(\t)$ strictly decreases.
\end{proof}

\begin{assumption}\label{as3}
	$s_1(\t)$ is strictly increasing in $\t$ for any $\tau\in (-x_u,x_q)$.
\end{assumption}

In other words, we suppose that for a fixed $\tau$, increases in $\t$ must lead to increases in $s_1(\t)$. This assumption states that applicants with stronger profiles should receive higher scores from an algorithm that predicts qualification, which must clearly be true in the economics sense. Moreover, this assumption is mathematically benign because $s_1(\t)$ is non-increasing only under unusual mathematical circumstances. We can rewrite $s_1(\t)$:
\begin{equation}\label{eq:g1}
s_1(\t) = \E[\k(\t, \Gamma) \mid \Gamma > \g_1(\t)],\quad \g_1(\t) \equiv \inf\left\{\g : \frac{h_q(\t, \g )}{h_u(\t, \g)} > \frac{(1-\pi)(x_u+\tau)}{\pi(x_q-\tau)}\right\}.
\end{equation}
By Assumption \ref{as1}, we know that the likelihood ratio is strictly increasing in $\t$ and $\g$ which means that $\g_1(\t)$ is strictly decreasing in $\t$. We also know that $\k(\t,\Gamma)$ is strictly increasing in $\t$ and $\g$. Thus, $s_1(\t)$ being non-increasing in $\t$ corresponds to the expectation of the posterior decreasing against competing forces of $\t$ increasing and the condition $\Gamma > \g_1(\t)$ admitting smaller observations of $\Gamma$. It is clear such a scenario would only occur under unrealistic assumptions imposed on the joint distribution of $\Theta$ and $\Gamma$.

\begin{proposition}\label{s1dec}
	As $m\to \infty$, $u(s_1(\tau))$ is almost surely continuously, strictly decreasing in $\tau$.
\end{proposition}
\begin{proof}
	See Appendix \hyperref[AppendixB]{B}.
\end{proof}

Proposition \ref{s1dec} states that the firm's expected profit from using $s_1$ strictly decreases in $\tau$. In other words, the predictor of qualification $Q$ makes the firm worse off when prejudice is high. Thus, a firm enjoys environments with low discrimination and suffers from prior eras of high discrimination. Even further, the firm benefits more from a data set that favors disadvantaged applicants, rather than one that is entirely impartial. 

\begin{proposition}\label{s2incdec}
	As $m\to \infty$, $u(s_2(\tau))$ is almost surely continuously strictly increasing, then strictly decreasing in $\tau$.
\end{proposition}
\begin{proof}
	See Appendix \hyperref[AppendixB]{B}.
\end{proof}

Similarly, Proposition \ref{s2incdec} reveals that the predictor of acceptance $A$ makes the firm worse off when prejudice is high. On the other hand, profits decline in environments where disadvantaged applicants are highly favored. Thus, the firm is better off with $s_2$ in environments without extreme levels of bias in either direction. Figure \ref{fig:fig4} is a conceptual illustration of the comparative statics of the expected profits of both algorithms over $\tau$. The vertical dotted line represents the point at which expected profit of $s_2$ begins to decrease in $\tau$.

The intuition behind Proposition \ref{s1dec} stems from the fact that $s_1$ assigns falsely high scores to more unqualified applicants as $\tau$ increases, i.e., ``bias reversal.'' This leads to the firm hiring more unqualified applicants when using $s_1$. By Corollary \ref{cor1}, we know that $\lim_{\tau \to -x_u}s_1 = \E[Q\mid \t]$, which is the best prediction of $Q$ given the observables. The accuracy of $s_1$ in this case allows the firm to make well-informed hires, leading to high expected profit. On the other hand, since $s_2$ predicts acceptance in the first-stage, extreme levels of $\tau$ in either direction lead to inaccurate predictions, explaining why the expected profit of $s_2$ drops off as $\tau$ tends towards either $-x_u$ or $x_q$.

	\begin{figure}[H]
	\centering
	\captionsetup{justification=centering}
	\resizebox{12cm}{!}{% \documentclass[12pt]{article}
% \linespread{1.25}
% \usepackage{pgfplots}
% \pgfplotsset{compat = newest}
% \usetikzlibrary{positioning, arrows.meta}
% \usepgfplotslibrary{fillbetween}
% \usepackage{amsmath}

% \begin{document}
\linespread{1.25}
\pgfplotsset{compat = newest}
\begin{tikzpicture}
\begin{axis}[
title = {(a) $\E[\post(s_1(\tau))]$},
xmin = -11,
xmax = 11,
ymin = -0.07,
ymax = 1,
xtick = {-10, 10},
ytick = {0, 1},
clip = false,
scale = 0.8,
x label style = {at={(axis description cs:0.5,-0.1)},anchor=north}, xlabel={Human Prejudice $\tau$},
y label style={anchor=south}, ylabel={Expected Profit of $s_1$},
xticklabels={$-x_u$,$x_q$},
yticklabels = {},
]
\addplot[domain = -10:10, restrict y to domain = -10:10, samples = 300, color = blue] {0.9 * (0.79 - 0.0286 * x  - 0.002988 * x^2 - 0.0000943 * x^3)};
\end{axis}
\end{tikzpicture}
\hspace{20pt}
\begin{tikzpicture}
\begin{axis}[
title = {(b) $\E[\post(s_2(\tau))]$},
xmin = -11, xmax = 11,
ymin = -0.07, ymax = 1,
xtick = {-10, 10},
ytick = {0, 1},
clip = false,
scale = 0.8,
x label style = {at={(axis description cs:0.5,-0.1)},anchor=north}, xlabel={Human Prejudice $\tau$},
y label style={anchor=south}, ylabel={Expected Profit of $s_2$},
xticklabels={$-x_u$,$x_q$},
yticklabels = {}
]
\draw[dashed] (1.3, 1) -- (1.3, -0.07);
\addplot[domain = -10:10, restrict y to domain = -10:10, samples = 300, color = red] {0.5 * (1.5 + 0.0186 * x  - 0.007388 * x^2 +  0.00043 * x^3)};
\end{axis}
\end{tikzpicture}
% \end{document}
}
	\caption{Expected profit of $s_1$ and $s_2$ over human prejudice $\tau$}\label{fig:fig4}
\end{figure}
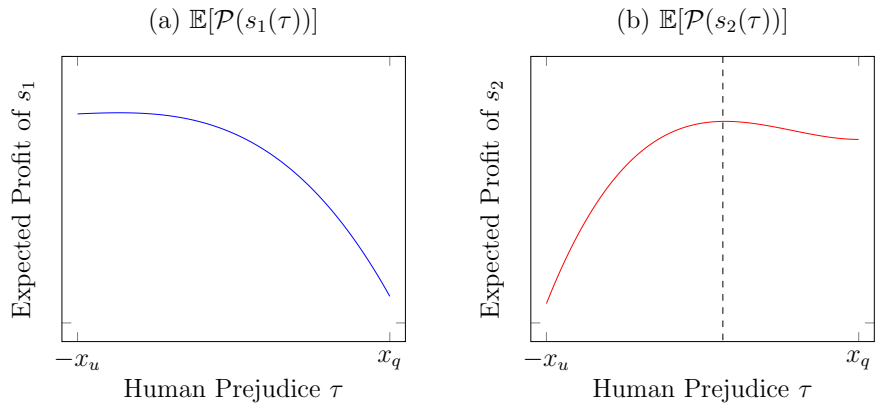

Now that the behaviors of $u(s_1(\tau))$ and $u(s_2(\tau)))$ are well understood, we compare the profits generated by the two algorithms. In the following, we directly discuss expected profit $\E[\post(s(\tau))]$ for convenience. Lemma \ref{slln} guarantees that average profit converges to expected profit almost surely.

\begin{proposition}\label{leftprofit}
	As $\tau\to -x_u$, the expected profit of $s_1$ is greater than that of $s_2$ if $v_u(1-\pi) > v_q\pi$.
\end{proposition}

\begin{proof}
	See Appendix \hyperref[AppendixB]{B}.
\end{proof}

\begin{proposition}\label{rightprofit}
	As $\tau\to x_q$, the expected profit of $s_2$ is greater than that of $s_1$ if and only if $v_u(1-\pi) > v_q\pi$.
\end{proposition}

\begin{proof}
	This is trivial once we compute $\lim_{\tau \to x_q}\E[\post(s_1(\tau))] = v_q\pi - v_u(1-\pi)$ and\\ $\lim_{\tau \to x_q}\E[\post(s_2(\tau))] = 0$.
\end{proof}

Proposition \ref{leftprofit} states that a firm from a low prior era of discrimination is better off with the predictor of qualification $s_1$ rather than the predictor of acceptance $s_2$. In other words, firms are better suited with ``bias reversal'' than ``bias in, bias out,'' when the data set maximally favors disadvantaged applicants. On the other end, a firm dealt a highly prejudiced data set is better off with $s_2$ than $s_1$.

The above results rely on the condition that $v_u(1-\pi)>v_q\pi$. This condition precisely means that the expected loss from hiring an unqualified applicant is greater than the expected profit from hiring a qualified applicant. This condition is equivalent to
\[
\E[Qv_q - (1 - Q)v_u] < 0,
\]
which rules out the pathological case where a firm assigning the baseline probability $\P(Q=1) = \pi$ to every applicant returns a positive expected profit. Thus, not only is this condition realistic, it allows us to focus on hiring environments where the firm would consign an algorithm to score applicants in the first place. Moreover, firms can fully observe whether or not this condition holds. For our final results, we add a mild regularity condition on the distribution of $\Theta$ to further rule out pathological cases.

\begin{defn}[Critical Signal]\label{critsig}
	\normalfont{The \textit{critical signal} $\underline{\t}$ of an algorithm $s(\tau)$ for some $\tau$ is the largest applicant signal that is rejected by} $s$:
	\[
	\underline{\t}(s(\tau)) := \inf\{\t : s(\t, \tau) > \b\}.
	\]
	\normalfont{We denote $\t_1(\tau) \equiv \underline{\t}(s_1(\tau))$ and $\t_2(\tau) \equiv \underline{\t}(s_2(\tau))$}.
\end{defn}

For example, if an applicant emits a signal $\t\in(\t_1,\t_2)$ for some $\tau$, this applicant is accepted by $s_1$ but rejected by $s_2$. Let $\tau_0$ denote the level of prejudice where $\E[\post(s_2(\tau))]$ is maximized, i.e., begins to decrease in Figure \ref{fig:fig4}.

\begin{assumption}\label{asmain}
	For $\tau \in (\tau_0, x_q)$, we have $\E[Q\mid \Theta \in (\t_1,\t_2)] < v_u/(v_u + v_q)$.
\end{assumption}

Intuitively, this assumption states that applicants who are accepted by $s_1$ and rejected by $s_2$ at high levels of $\tau$ are not overly likely to be qualified. Recall that by Theorem \ref{bin}, for high $\tau$, $s_1$ inflates scores while $s_2$ deflates scores. An applicant with a better signal $\t$ is able to more strongly resist the ``bias in, bias out'' quality exhibited by $s_2$, precisely because $s_2$ is increasing in $\t$. Assumption \ref{asmain} addresses the applicants with signals that have weaker $\t$, thus being rejected by $s_2$---these low signal-valued applicants would naturally have a reduced probability of being qualified. In particular, if $v_u(1-\pi) > v_q\pi$, then we have that $\pi < v_u/(v_u + v_q)$, which further rationalizes our assumption. To drive this point home, we show that Assumption \ref{asmain} holds at the limiting values of $\tau$, suggesting that any violations of this assumption would require unnatural behavior in the comparative statics of $\E[Q\mid \Theta \in (\t_1,\t_2)]$ over $\tau$.

\begin{lemma}\label{justify}
	If $v_u(1-\pi) > v_q\pi$, as $\tau \to \tau_0$ and $\tau\to x_q$, Assumption \ref{asmain} holds.
\end{lemma}

\begin{proof}
	See Appendix \hyperref[AppendixB]{B}.
\end{proof}

Agreeing with the intuition behind Assumption \ref{asmain}, all simulations conducted in Section \ref{sec:5} using the most common distributions with the MLRP (Assumption \ref{as1}) satisfy the assumption. In particular, the author is unaware of any continuous distributions that violate Assumption \ref{asmain}---this leads us to believe the assumption is true under very natural, yet specific parametric conditions. It is also entirely possible that the assumption is simply always true without any restrictions at all.

\begin{theorem}\label{main2}
	If $v_u(1-\pi) > v_q\pi$, there exists a unique $\underline{\tau}\in (-x_u,x_q)$ such that expected profits of $s_1$ and $s_2$ are equal.
\end{theorem}

\begin{proof}
	See Appendix \hyperref[AppendixB]{B}.
\end{proof}

We have shown that the expected profit from $s_1$ and $s_2$ are equal at only one level of prejudice. Moreover, Propositions \ref{leftprofit} and \ref{rightprofit} state that expected profit from $s_1$ is greater than $s_2$ as $\tau \to -x_u$ and from $s_1$ is less than $s_2$ as $\tau \to x_q$. Thus, we can determine which algorithm yields higher profits for any $\tau\in (-x_u, x_q)$ according to an immediate application of Theorem \ref{main2}:

\begin{corollary}\label{gprofit}
	Suppose $v_u(1-\pi) > v_q\pi$ and let $\tau\in (-x_u, x_q)$. The expected profit of $s_1$ is greater than that of $s_2$ if and only if $\tau < \ut$.
\end{corollary}

\begin{proof}
	The proof immediately follows from Theorem \label{main} and monotonicity of expected profits in $(-x_u, \tau_0)$.
\end{proof}

In other words, if the prior era was sufficiently prejudiced, then the firm receives greater profits from $s_2$. If disadvantaged applicants were sufficiently favored in the prior era, then the firm  receives greater profits from $s_1$. Figure \ref{fig:fig5} displays a complete, conceptual picture of the expected profits of $s_1$ and $s_2$. The red curve corresponds to the expected profit of $s_1$, and the blue corresponds to that of $s_2$. The black dot represents the unique point where profits are equal as stated in Theorem \ref{main2}.

\begin{figure}[H]
	\centering
	\captionsetup{justification=centering}
	\resizebox{9cm}{!}{% \documentclass[12pt]{article}
% \linespread{1.25}
% \usepackage{pgfplots}
% \pgfplotsset{compat = newest}
% \usetikzlibrary{positioning, arrows.meta}
% \usepgfplotslibrary{fillbetween}
% \usepackage{amsmath}

% \begin{document}
\linespread{1.25}
\pgfplotsset{compat = newest}
\begin{tikzpicture}
\begin{axis}[
xmin = -11,
xmax = 11,
ymin = -0.07,
ymax = 1,
xtick = {-10, 10},
ytick = {0, 1},
clip = false,
scale = 0.8,
x label style = {at={(axis description cs:0.5,-0.1)},anchor=north}, xlabel={Human Prejudice $\tau$},
y label style={anchor=south}, ylabel={Expected Profit},
xticklabels={$-x_u$,$x_q$},
yticklabels = {},
]
\addplot[domain = -10:10, restrict y to domain = -10:10, samples = 300, color = red] {0.9 * (0.79 - 0.0286 * x  - 0.002988 * x^2 - 0.0000943 * x^3)};
\addplot[domain = -10:10, restrict y to domain = -10:10, samples = 300, color = blue] {0.5 * (1.55 + 0.0186 * x  - 0.007388 * x^2 +  0.00043 * x^3)};
\addplot[color = black, mark = *, only marks, mark size = 2pt] coordinates {(-1.53, 0.745)};
\draw[dashed] (-1.54, 0.745) -- (-1.54, -0.07);
\draw[dashed] (-11, 0.1) -- (11, 0.1);
\draw[dashed] (-11, 0.7) -- (11, 0.7);
\node [right] at (11, 0.7) { $0$};
\node [below] at (-1.54, -0.07) { $\ut$};
\node [right] at (11, 0.1) { $v_q\pi - v_u(1-\pi)$};
\draw[->]  (-6.5, 0.86) to (-7, 0.8);
\draw[->] (6.3, 0.83) to (7, 0.75);
\node [above] at (6.3, 0.83) {$s_2$};
\node [above] at (-6.5, 0.86) {$s_1$};
\end{axis}
\end{tikzpicture}
% \end{document}
}
	\caption{Comparison of expected profit of $s_1$ and $s_2$ over possible levels of prejudice $\tau$}\label{fig:fig5}
\end{figure}
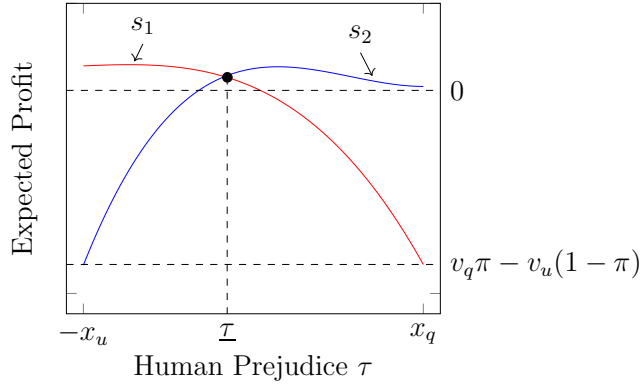

The above results demonstrate that firms endowed with a highly prejudice data set obtain higher profits from $s_2$, the algorithm that predicts acceptance $A$ and exhibits ``bias in, bias out.'' If the data set instead favors disadvantaged applicants, the firm receives higher profits from $s_1$, the algorithm that predicts qualification $Q$ and exhibits ``bias reversal.'' Interestingly, we notice that the results of Theorem \ref{main2} and Corollary \ref{gprofit} reflect the opposite trend of Proposition \ref{prop1} and Corollary \ref{greater}. These conflicting results reveal a deeper tension between the human decision-maker's taste-based discrimination in the first stage and the firm's profit in the second stage. Higher levels of prejudice in the first stage lead to $s_1$ reversing bias thereby favoring disadvantaged applicants; however, firm profits decline. Lower levels of prejudice in the first stage inflated scores of $s_2$, but, similarly, firm profits decline.

	\section{Simulation}\label{sec:5}
	
	In this section, we conduct a series of simulations emulating the first and second stage of our model. We display results using some of the most common distributions for $\Theta, \Gamma$ satisfying the monotone likelihood ratio property. We visualize the claims of Theorem \ref{main2} and Corollary \ref{gprofit}, and show that over the interval $\tau \in (-x_u,x_q)$, the expected profits intersect at a unique point, as guaranteed.
	
	\begin{figure}[H]
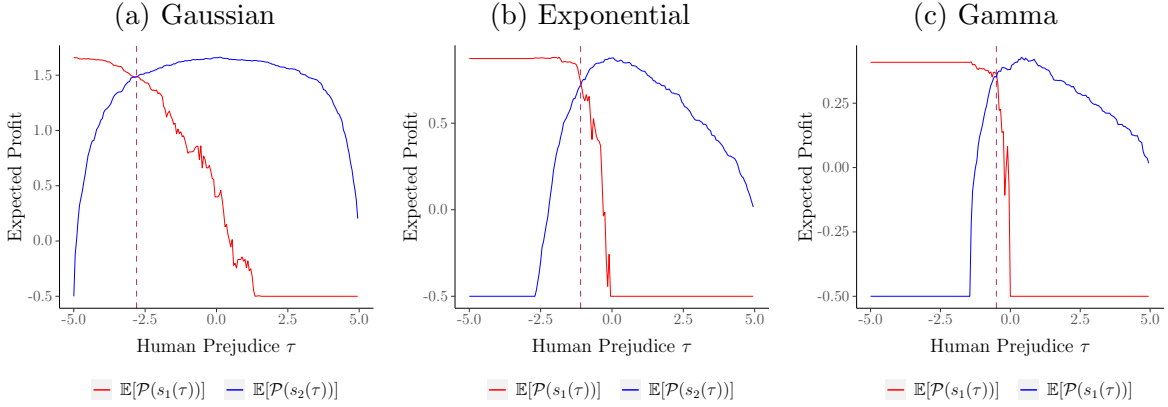

		\centering
		\begin{subfigure}[Figure A]{0.31\linewidth}
			\caption{Gaussian}
			\centering
			\resizebox{\linewidth}{!}{\input{Plots/normal_prof.tex}}
			\label{figa}
		\end{subfigure}
		\begin{subfigure}[Figure B]{0.31\linewidth}
			\caption{Exponential}
			\centering
			\resizebox{\linewidth}{!}{\input{Plots/exp_prof.tex}}
			\label{figb}
		\end{subfigure}
		\begin{subfigure}[Figure C]{0.31\linewidth}
			\caption{Gamma}
			\centering
			\resizebox{\linewidth}{!}{\input{Plots/gamma_prof.tex}}
			\label{figc}
		\end{subfigure}
		\captionsetup{justification=centering}
		\caption{Expected profits of $s_1$ and $s_2$ over $\tau$ for\\ three different distributions satisfying MLRP}
		\label{fig:fig6}
	\end{figure}
	
	All figures in Figure \ref{fig:fig6} are simulation-generated with $n= 1000$, $m = 5000$, $\pi = 0.5$, $x_q = x_u = 5$, $v_q = 5$, and $v_u=6$. In Figure \ref{figa}, we choose $(\Theta\mid Q=1)\sim N(4,1)$ and $(\Theta \mid Q=0 )\sim N(2,1)$. In Figure \ref{figb}, we have $(\Theta\mid Q=1)\sim \text{exp}(1/3)$ and $(\Theta\mid Q=0)\sim \text{exp}(1)$. In Figure \ref{figc}, we have $(\Theta\mid Q=1)\sim \Gamma(3, 1/3)$ and $(\Theta\mid Q=0)\sim \Gamma(3, 1)$.\footnote{The parameters of the Gaussian distribution are mean and variance, respectively. The parameter of the exponential distribution is rate. The parameters of the Gamma distribution are shape and rate, respectively.} The distribution of $\Gamma$ (which is correlated with $\Theta$) is chosen identically. We can indeed see that for all three distributions, the uniqueness of $\ut$ holds over $\tau \in (-x_u,x_q)$, which is represented by the vertical dotted line. Next, we display the same plots using the discrete Binomial and Poisson distributions.  Although the theoretical guarantees in this paper hold exclusively for continuous distributions, we find that simulations reflect similar outcomes with discrete distributions. This suggests that our results represent a deeper shadow of a larger phenomenon.
	
	\begin{figure}[H]
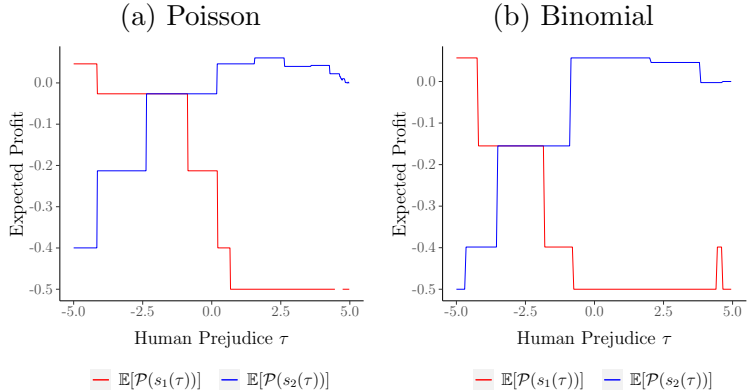

		\centering
		\begin{subfigure}[Figure A]{0.3\linewidth}
			\caption{Poisson}
			\centering
			\resizebox{\linewidth}{!}{\input{Plots/poi_prof.tex}}
			\label{figa2}
		\end{subfigure}
		\begin{subfigure}[Figure B]{0.3\linewidth}
			\caption{Binomial}
			\centering
			\resizebox{\linewidth}{!}{\input{Plots/bin_prof.tex}}
			\label{figb2}
		\end{subfigure}
		\captionsetup{justification=centering}
		\caption{Expected profits of $s_1$ and $s_2$ over $\tau$ for\\ discrete distributions satisfying MLRP}
		\label{fig:fig7}
	\end{figure}
	
	In Figure \ref{figa2}, we take $(\Theta \mid Q=1)\sim \text{Po}(3)$ and $(\Theta \mid Q=0)\sim \text{Po}(1)$. In Figure \ref{figb2}, we have $(\Theta \mid Q=1)\sim B(5, 0.5)$ and $(\Theta\mid Q=0)\sim B(5, 0.3).$\footnote{The first parameter corresponds to number of trials and the second parameter is success probability for each trial.} All other parameters are the same as Figure \ref{fig:fig6}. In these examples, expected profits are not strictly monotonic for either algorithms, but it appears that expected profit of $s_1$ is weakly decreasing while that of $s_2$ is weakly increasing then decreasing---these traits are analogous to Proposition \ref{s1dec} and \ref{s2incdec}. Similar trends appear for our main results. We emphasize that our results are not guaranteed for all discrete distributions.
	
	\section{Closing Remarks}\label{sec:6}
	
	In this paper, we constructed a rich two-stage model that incorporates a statistical discriminating firm hiring applicants through the use of an algorithm. We formalized average profits and compared the firm's level of success using either $s_1$ or $s_2$. We prove the almost sure existence of a unique point where average profit from $s_1$ and $s_2$ are equal as time goes on: thus, firms using highly prejudiced data sets find more success with $s_2$ over $s_1$ on average, while firms using data sets that greatly favor disadvantaged applicants are more successful with $s_1$ over $s_2$. Along the way, we showed that the expected profit of $s_1$ is strictly decreasing in $\tau$; the expected profit of $s_2$ is strictly increasing, then decreasing in $\tau$. These results reveal important considerations that firms must make when using algorithms that are weighed down by a prejudiced training data set from a previous time period. Even if a firm is completely impartial, the algorithm that maximizes profits is influenced by the level of prejudice embedded in the data. In our case, firms using highly prejudiced data find more profit from the algorithm that exhibits ``bias in, bias out.'' On the other hand, firms using data that favors disadvantaged applicants are more profitable with the algorithm that exhibits ``bias reversal.''
	
	 As an important extension, one may consider competitive environments where firms compete by choosing algorithms that maximize profits. In such a scenario, one may ask what type of algorithm is used in equilibrium, depending on the level of prejudice embedded in that industry's prior era. More generally, if there are $n$ firms, one may allow disadvantaged applicants to apply to different firms depending on perceived levels of prejudice. If each applicant observes a noisy signal for each of the firm's data set's prejudice $\tau_i$ for each $i\in [n]$, applicants will naturally apply to the firm that maximizes their individual utility. Knowing this, firms determine their algorithm of choice in response to the applicants' decisions in order to maximize expected profits. Even further, we may allow the firms to choose from a large space of algorithms in order to maximize expected profits---this greater generality would better inform the nature of the algorithm used in competitive equilibrium among the firms. 
	
	\section*{Appendix}
	
	\subsection*{Appendix A: Proofs from Section \ref{sec:3}}\label{AppendixA}
	
	\begin{proof}[Proof of Lemma \ref{cont}]
		We may write
		\begin{equation}\label{eq-dag}
			\E[Q\mid \Theta= \t, A=1] = \frac{\P(A=1\mid Q=1,\Theta =\t)\P(Q=1\mid \Theta=\t)}{\P(A=1\mid \Theta = \t)}.
		\end{equation}
		We see that $\P(Q=1\mid \Theta=\t)$ is trivially continuous in $\tau$, since it is constant.
		\begin{equation}\label{eq-ast}
			\P(A=1\mid \Theta = \t) = \P(\Gamma > \g_1(\tau)\mid \Theta = \t) = 1 - \P(\Gamma\leq \g_1(\tau)\mid \Theta =\t).
		\end{equation}
		
		It's easy to see that $\g_1(\tau)$ is continuous in $\tau$ since the right-hand side of the inequality in (\ref{eq:g1}) is continuous in $\tau\in (-x_u,x_q)$, and the left-hand side is continuous in $\g$ by Assumption \ref{as1}. Then, since the right-hand term is a CDF, and $\Gamma$ is a continuous random variable, it takes values in the open interval $(0,1)$ and is continuous. In particular, (\ref{eq-ast}) is never 0 for all $\t$, and so the right-hand side of (\ref{eq-dag}) is well-defined. Recall that (\ref{eq-ast}) is equal to $s_2$, so $s_2$ is continuous in $\tau$. We can identically state that $\P(A=1\mid Q=1, \Theta=\t)$ is continuous in $\tau$. Thus, $s_1$ is the ratio of two continuous functions with non-zero denominator, and we conclude that $s_1$ is continuous in $\tau$.
	\end{proof}

	\begin{proof}[Proof of Proposition \ref{prop1}]
		Recall that we may write $s_2$ as follows.
		\[
		\E[A\mid \Theta = \t] = 1- \P(\Gamma\leq \g_1(\tau)\mid \Theta = \t).
		\]
		Then we have:
		\[
		\lim_{\tau\to -x_{u}} \frac{(1-\pi)(x_u+\tau)}{\pi(x_q-\tau)} = 0.
		\]
		This implies that $\lim_{\tau\to -x_{u}}\g_1(\tau) = -\infty.$ Thus:
		\[
		\lim_{\tau\to -x_{u}} 1 - \P(\Gamma\leq \g_1(\tau) \mid \Theta =\t) = 1.
		\]
		Recall that $s_1(\t, \tau)$ is given by:
		\[
		s_1(\t, \tau) = \frac{\P(A=1\mid Q=1,\Theta =\t)\P(Q=1\mid \Theta=\t)}{\P(A=1\mid \Theta = \t)}.
		\]
		It's easy to see that as $\tau \to -x_{u}$, we have
		\[
		\P(A=1\mid Q=1,\Theta = \t) = \P(A=1\mid \Theta = \t) = 1.
		\]
		It follows that
		\begin{align*}
			\lim_{\tau\to -x_{u}} s_1(\t, \tau) &= \P(Q=1\mid \Theta = \t)
		\end{align*}
		Next, we evaluate the limit of $s_2$ as $\tau\to x_{q}$. 
		\[
		\lim_{\tau\to x_{q}} 1 - \P(\Gamma\leq \g_1(\tau) \mid \Theta = \t) = 0.
		\]
		For $s_1$, noting that $A=1\iff \Gamma > \g_1(\tau)$, by the law of iterated expectation we get:
		\begin{align*}
			s_1(\t, \tau) &= \E[Q\mid \Theta =\t, \Gamma > \g_1(\tau)]\\
			&= \E[\E[Q\mid \Theta = \t, \Gamma]\mid \Gamma > \g_1(\tau)]\\
			&= \E[\k(\t, \Gamma) \mid \Gamma > \g_1(\tau)]
		\end{align*}
		The condition $A=1$ is equivalent to $\k(\t, \g) x_q - (1-\k(\t, \g))x_u > \tau.$ For $\tau\to x_q$, the accepted applicants are such that:
		\[
		\k(\t, \g)x_q - (1-\k(\t, \g))x_u > x_q\implies \k(\t, \g) = 1.
		\]
		Therefore, we have
		\[
		\lim_{\tau\to x_{q}}s_1(\t, \tau) = 1.
		\]
		We have shown that $s_2$ strictly decreases from 1 to 0 over the interval $(-x_u,x_q)$, and $s_1$ strictly increases from $\P(Q=1\mid \Theta=\t)$ to 1 over the interval $(-x_u,x_q)$. By Lemma \ref{cont}, we know that both are continuous, which proves existence. Uniqueness trivially holds by Theorem \ref{bin}.
	\end{proof}
	
	\subsection*{Appendix B: Proofs from Section \ref{sec:4}}\label{AppendixB}
	
	\begin{customdefn}{B.1}[Critical Score]
		\normalfont{We define the \textit{critical score} $\beta$ the smallest score needed for an applicant to be accepted.}
		\[
		\b := \inf\{s: s(\tau)v_q - (1-s(\tau))v_u > 0\}.
		\]
	\end{customdefn}
	
	When $s=\beta$, we have $\mathcal{N}(s_i) = 0$. It's easy to see that $\mathcal{N}(s_i)$ is strictly increasing in $s$, and $A'(s_i)$ is weakly increasing in $s$. 
	
	\begin{proof}[Proof of Proposition \ref{s1dec}]
		
		\begin{customlemma}{B.3}\label{slln}
		For $s\in \{s_1,s_2\}$, we have that $u(s(\tau)) \to \E[\post(s(\tau))]$ almost surely.
		\end{customlemma}
	
		\begin{proof}
			In order to apply the strong law of large numbers, it suffices to show that $\E|\post(s(\tau))| < \infty$, which is clearly true since $\post(s(\tau)) \in \{0, x_q, -x_u\}$.
		\end{proof}
		
		Applying Lemma \ref{slln}, we investigate $\E[\post(s_1(\tau))]$. Note that $\post(s_1(\tau)) = 0$ if $s_1 < \b$ because $A' = 0$.
		\begin{align*}
		\E[\post(s_1(\tau))] &= \P(s_1 > \b)\E[\post(s_1(\tau)) \mid s_1 > \b] + \P(s_1 \leq \b)\E[\post(s_1(\tau)) \mid s_1 \leq \b]\\
		&= \P(s_1 > \b)\E[Qv_q - (1-Q)v_u \mid s_1 > \b]\\
		&= \P(\Theta > \t_1(\tau))(v_q + v_u)\E[Q\mid \Theta > \t_1(\tau)] - v_u\P(\Theta > \t_1(\tau))\\
		&= (v_q+v_u)\int_{\t_1(\tau)}^\infty \E[Q\mid \Theta = \t]f(\t)d\t - v_u\P(\Theta > \t_1(\tau)).
		\end{align*}
	The first line follows from the law of total expectation. The third line follows upon recognizing that $s_1^{-1}(\b) = \t_1(\tau)$ where $\t_1(\tau)$ is defined in Definition \ref{critsig}. The last line follows from the law of iterated expectation. We know that $\t_1(\tau)$ is strictly decreasing, so it is differentiable almost everywhere by Lebesgue's theorem on monotone functions. The derivative is also strictly negative almost everywhere. Let $\tau$ be such a point. Then:
	\begin{gather*}
		\frac{\partial}{\partial \tau} (v_q+v_u)\int_{\t_1(\tau)}^\infty \E[Q\mid \Theta = \t]f(\t)d\t - v_u\P(\Theta > \t_1(\tau))\\
		=\frac{\partial}{\partial \tau} (v_q+v_u)\left(\int_{\t_1(\tau)}^A \E[Q\mid \Theta = \t]f(\t)d\t + \int_{A}^\infty \E[Q\mid \Theta = \t]f(\t)d\t  \right) - v_u\P(\Theta > \t_1(\tau)) \\
		=-(v_q + v_u)\phi(\t_1)f(\t_1)\t_1' + v_uf(\t_1)\t_1'
	\end{gather*}
where $A$ is some finite constant greater than $\t_1(\tau)$. The last line follows from Leibniz's integral rule. The derivative is negative if and only if:
\begin{align*}
	(v_q + v_u)\phi(\t_1)f(\t_1)\t_1' > v_uf(\t_1)\t_1' &\iff (v_q + v_u)\phi(\t_1) < v_u\\
	&\iff \phi(\t_1) < \b.
\end{align*}
We know that $\phi(\t_1)$ is strictly increasing in $\t_1$ by Assumption \ref{as1}, so the supremum of $\phi(\t_1)$ occurs when $\tau \to -x_u.$ Recall that $\t_1 = \inf\{\t : s_1 > \b\}$, so $\lim_{\tau \to -x_u} \t_1 = \inf\{\t : \phi(\t) > \b\}$. Thus, $\lim_{\tau \to -x_u}\phi(\t_1) = \b$. Thus, for any $\tau > -x_u$, the derivative of $\E[\post(s_1(\tau))]$ is strictly negative, when the derivative of $\t_1(\tau)$ exists and is non-zero.

The measure space for $\tau$ is given by $((-x_u, x_q), \mathbf{B}, \mathfrak{m})$ where $\mathbf{B}$ is the Borel $\sigma$-field and $\mathfrak{m}$ is the Lebesgue measure. Let $C$ and $D$ be the sets where $\t_1'$ exists and is non-zero, respectively. We know that $\mathfrak{m}(C) = \mathfrak{m}(D) = 0$. The set on which the derivative of $\E[\post(s_1(\tau))]$ exists is $C\cup D$. By countable subadditivity, we have that $\mathfrak{m}(C \cup D)\leq \mathfrak{m}(C) + \mathfrak{m}(D) = 0$. Thus, $\mathfrak{m}(C\cup D) = 0$, so the derivative of $\E[\post(s_1(\tau))]$ exists almost everywhere and is negative almost everywhere. It is well known that a function with negative derivative almost everywhere is strictly decreasing. 
	\end{proof}

\begin{proof}[Proof of Proposition \ref{s2incdec}]
	The proof is similar to that of Proposition \ref{s1dec}. Since $\t_2(\tau)$ is strictly increasing in $\tau$, it is differentiable almost everywhere. The derivative is also strictly positive almost everywhere. Let $\tau$ be such a point and we have:
	\begin{align*}
		\frac{\partial}{\partial \tau} \E[\post(s_2(\tau))] &=  \frac{\partial}{\partial \tau} (v_q+v_u)\int_{\t_2(\tau)}^\infty \E[Q\mid \Theta = \t]f(\t)d\t - v_u\P(\Theta > \t_2(\tau))\\
			&=-(v_q + v_u)\phi(\t_2)f(\t_2)\t_2' + v_uf(\t_2)\t_2'.
	\end{align*}
This is positive if and only if:
\begin{align*}
v_uf(\t_2)\t_2' > (v_q + v_u)\phi(\t_2)f(\t_2)\t_2' &\iff \b > \phi(\t_2).
\end{align*}
We know that $\phi(\t_2)$ is strictly increasing in $\tau$ because $\t_2$ is strictly increasing in $\tau$. Taking limits:
\[
\lim_{\tau \to -x_u} \phi(\t_2) = 0,\quad \lim_{\tau \to x_q} \phi(\t_2) = 1.
\]
Since $\b\in (0,1)$, we conclude that there exists a unique point $\delta$ such that $\E[\post(s_2(\tau))]$ has positive derivative almost everywhere on $(-x_u, \delta)$ and negative derivative almost everywhere on $(\delta, x_q)$. Using an identical argument to the Proof of Proposition \ref{s1dec}, we conclude that $\E[\post(s_2(\tau))]$ strictly increasing on $(-x_u, \delta)$ and strictly decreasing on $(\delta, x_q)$.
\end{proof}
	
	\begin{proof}[Proof of Proposition \ref{leftprofit}]
		For convenience, denote $\phi(\t) := \P(Q=1\mid \Theta = \t)$.
		\begin{align*}
		\E[\post(s(\tau))] &= v_q\P(A'=1,Q=1) - v_u\P(A'=1, Q=0)\\
			&= v_q\pi \P(s>\b | Q=1) - v_u(1-\pi)\P(s>\b\mid Q=0).
		\end{align*}
		Then we compute:
		\begin{align*}
			\lim_{\tau \to -x_u} \E[\post(s_1(\tau))] &= v_q\pi \P(\phi(\Theta) > \b \mid Q=1) - v_u(1-\pi)\P(\phi(\Theta) > \b\mid Q=0)\\
			\lim_{\tau \to -x_u}\E[\post(s_2(\tau))] &= v_q\pi - v_u(1-\pi)
		\end{align*}
		by Corollary \ref{cor1}. Proceed as follows.
		\begin{align*}
			v_q\pi \P(\phi(\Theta) > \b \mid Q=1) - v_u(1-\pi)\P(\phi(\Theta) > \b \mid Q=0) > v_q\pi - v_u(1-\pi) &\iff\\
			v_q\pi(\P(\phi(\Theta) > \b\mid Q=1) - 1) > v_u(1-\pi)(\P(\phi(\Theta) > \b \mid Q=0) - 1) &\iff\\
			\frac{\pi}{1-\pi}\left(\frac{\P(\phi(\Theta) > \b\mid Q=1) - 1}{\P(\phi(\Theta) > \b \mid Q=0) - 1}\right) < \frac{v_u}{v_q} &\iff\\
			\frac{v_u}{v_q} > \frac{\pi}{1-\pi}\left(\frac{\P(\phi(\Theta) \leq \b\mid Q=1)}{\P(\phi(\Theta) \leq \b \mid Q=0)}\right).
		\end{align*}
		Since $v_u/v_q > \pi/(1-\pi)$, all that is left to show is that the term in parentheses is less than or equal to 1. Denote $B(\Theta) = \1\{\phi(\Theta) \leq \b\}$. We know that $\phi(\t)$ is strictly increasing in $\t$ from Assumption \ref{as1} (in particular, $f_q(\t)/f_u(\t)$ satisfies the strict MLRP). Thus, $B(\t)$ is weakly decreasing in $\t$. Let $\mu$ denote the conditional law of $\Theta$ with $Q=1$ and $\nu$ the conditional law of $\Theta$ with $Q=0$. That is, $\mu$ corresponds to the CDF $F_q(\t)$ and $\nu$ corresponds to $F_u(\t)$. By the strict MLRP, we have $F_q(\t) < F_u(\t)$ for all $\t$. In other words, $\mu$ has strict first-order stochastic dominance over $\nu$. It is well known that this equivalently means:
		\[
		\E_{\Theta \sim \mu}[B(\Theta)] \leq \E_{\Theta \sim \nu}[B(\Theta)]
		\]
		for non-increasing $B(\t)$. But this is equivalent to:
		\[
		\P(\phi(\Theta) \leq \b \mid Q=1) \leq \P(\phi(\Theta) \leq \b \mid Q=0).
		\]
		This concludes.
	\end{proof}

\begin{proof}[Proof of Lemma \ref{justify}]
	This Lemma is true as $\tau\to \tau_0$ from an immediate application from the following fact:
	
	\begin{customlemma}{B.4}\label{inter}
		As $\tau \to \tau_0$, the expected profit of $s_2$ is greater than that of $s_1$ if $v_u(1- \pi) > v_q\pi$.
	\end{customlemma}
	
	\begin{proof}[Proof of Lemma \ref{inter}]
		From Proposition \ref{s1dec} and \ref{s2incdec}, we know that expected profit of $s_1$ and $s_2$ are strictly decreasing and increasing in $\tau\in (-x_u,\tau_0)$. Thus, if there exists a $\tau\in (-x_u,\tau_0)$ such that expected profits are equal, it is unique in that interval, and by monotonicity, the claim holds. Recall that $\phi(\t_2(\tau_0)) = \b$ and $\phi(\t_1(\tau)) < \b$ for all $\tau$. Since $\phi$ is monotone, we have that $\t_2(\tau_0) > \t_1(\tau_0)$. From the proofs of Proposition \ref{s1dec} and \ref{s2incdec}, we showed that $\phi(\t_1(\tau)) = \b$ and $\phi(\t_2(\tau)) = 0$ as $\tau \to -x_u$, so $\t_1(\tau) > \t_2(\tau)$ as $\tau \to -x_u$. Since $\t_1(\tau)$ and $\t_2(\tau)$ are strictly decreasing and increasing in $\tau$, we conclude there exists a unique point where $\t_1(\tau) = \t_2(\tau)$ for $\tau \in (-x_u, \tau_0)$. Proceed as follows:
		\begin{align*}
			\E[\post(s_1(\tau))] = \E[\post(s_2(\tau))] &\iff\\
			v_q\pi \P(s_1>\b | Q=1) - v_u(1-\pi)\P(s_1>\b\mid Q=0) &=\\ v_q\pi \P(s_2>\b | Q=1) - v_u(1-\pi)\P(s_2>\b\mid Q=0) &\iff\\
			v_q\pi \P(\Theta>\t_1(\tau) | Q=1) - v_u(1-\pi)\P(\Theta>\t_1(\tau)\mid Q=0) &=\\ v_q\pi \P(\Theta>\t_2(\tau) | Q=1) - v_u(1-\pi)\P(\Theta > \t_2(\tau)\mid Q=0) &\iff\\
			v_q\pi(1-F_q(\t_1)) - v_u(1-\pi)(1 - F_u(\t_1)) &=\\ v_q\pi(1-F_q(\t_2)) - v_u(1-\pi)(1 - F_u(\t_2)) &\iff\\
			v_q\pi(F_q(\t_2) - F_q(\t_1)) = v_u(1-\pi)(F_u(\t_2) - F_u(\t_1)).
		\end{align*}
		When $\t_1(\tau) = \t_2(\tau)$, the equality holds. This concludes.
	\end{proof}
	Thus, at $\tau_0$, we have $\t_2 > \t_1$ and
	\[
	\E[\post(s_2(\tau_0))] > \E[\post(s_1(\tau_0))].
	\]
For arbitrary algorithm $s$, we have
	\begin{gather*}
		\E[\post(s(\tau_0))]  =\P(\Theta < \t_1)\E[\post(s(\tau_0))\mid \Theta < \t_1]+\\ \P(\Theta \in (\t_1,\t_2))\E[\post(s(\tau_0))\mid \Theta \in (\t_1,\t_2)] +
		\P(\Theta > \t_2)\E[\post(s(\tau_0))\mid \Theta > \t_2].
	\end{gather*}
If $\t < \t_1(\tau_0)$, we have that $A'(s_1(\tau_0)) = A'(s_2(\tau_0)) = 0$. Thus:
\[
\E[\post(s_1(\tau_0))\mid \Theta < \t_1] = \E[\post(s_2(\tau_0))\mid \Theta < \t_1] = 0.
\]
If $\t > \t_2(\tau_0)$, then $A'(s_1(\tau_0)) = A'(s_2(\tau_0)) = 1$. Thus:
\[
\E[\post(s_1(\tau_0))\mid \Theta > \t_2] = \E[\post(s_2(\tau_0))\mid \Theta > \t_2] = (v_q + v_u)\pi - v_u.
\]
Putting these facts together, we have
\[
\E[\post(s_2(\tau_0)) - \post(s_1(\tau_0))] >0 \iff \E[\post(s_2(\tau_0)) - \post(s_1(\tau_0))\mid \Theta \in (\t_1,\t_2)] > 0.
\]
Since $A'(s_2(\tau_0))=0$, we have $\post(s_2(\tau_0)) = 0$:
\[
-\E[\post(s_1(\tau_0))\mid \Theta \in (\t_1,\t_2)] > 0 \iff \E[Q\mid \Theta \in (\t_1,\t_2)] < \b.
\]
This concludes. For $\tau \to x_q,$ it suffices to take the limit:
\[
\lim_{\tau \to x_q} \E[Q\mid \Theta \in (\t_1,\t_2)] = \E[Q\mid \Theta \in (-\infty, \infty)] = \pi < \b.
\]
\end{proof}

\begin{proof}[Proof of Theorem \ref{main2}]
	
By Propositions \ref{s1dec} and \ref{s2incdec}, we have that $\E[\post(s_1(\tau))] - \E[\post(s_2(\tau))]$ is continuous. Proposition \ref{leftprofit} and Proposition \ref{rightprofit} tells us that $\lim_{\tau \to -x_u} \E[\post(s_1(\tau))] - \E[\post(s_2(\tau))] > 0$ and $\lim_{\tau \to x_q} \E[\post(s_1(\tau))] - \E[\post(s_2(\tau))]<0$ respectively. By the intermediate value theorem, there exists a point $\tau$ where $\E[\post(s_1(\tau))]=\E[\post(s_2(\tau))]$.

\begin{customlemma}{C.6}
	If $v_u(1- \pi) > v_q\pi$, there exists a unique $\ut\in (-x_u, \tau_0)$ such that expected profits of $s_1$ and $s_2$ are equal in this interval.
\end{customlemma}

\begin{proof}
	This immediately follows from Lemma \ref{inter} and Proposition \ref{s1dec} and \ref{s2incdec}.
\end{proof}

		All that remains is to show that expected profits do not intersect in the interval $(\tau_0, x_q)$. In particular, we wish to show $\E[\post(s_2(\tau))] > \E[\post(s_1(\tau))]$ for $\tau \in (\tau_0,x_q).$ From the proof of Lemma \ref{justify}, we can see that this is implied by
		\[
		\E[Q\mid \Theta \in (\t_1,\t_2)] < \b
		\]
		which is given by Assumption \ref{asmain}.
	\end{proof}
	\nocite{*}
	\bibliographystyle{aea}
	\bibliography{bib_2.bib}

\end{document}